\documentclass[11pt]{article}

\usepackage[american]{babel}
\usepackage[T1]{fontenc}
\usepackage[utf8]{inputenc}
\usepackage[text={15.0cm,22.0cm}]{geometry}
\usepackage{lmodern}
\usepackage{microtype}

\usepackage{amsmath, amssymb, amsthm, mathtools, thmtools}
\usepackage{wrapfig}
\usepackage{nicefrac}

\usepackage[inline]{enumitem}
\setlist[enumerate]{label=(\roman*),itemsep=-1\parsep,topsep=0.5\parsep}

\usepackage{hyperref}

\usepackage{authblk}

\declaretheorem{theorem}
\declaretheorem[sibling=theorem]{lemma}

\theoremstyle{definition}
\declaretheorem[sibling=theorem]{definition}

\theoremstyle{remark}
\declaretheorem[sibling=theorem]{remark}

\newcommand{\Cmax}{C_{\max}}
\newcommand{\RCm}{R||\Cmax}
\newcommand{\PCm}{P||\Cmax}
\newcommand{\RA}{P|\validmachs{j}|\Cmax}
\newcommand{\machs}{\mathcal{M}}
\newcommand{\jobs}{\mathcal{J}}
\newcommand{\validmachs}[1]{M(#1)}
\newcommand{\validjobs}[1]{J(#1)}

\newcommand{\treewidth}{\mathrm{tw}}

\newcommand{\amachs}[1]{M_{#1}}
\newcommand{\iamachs}[1]{\check{M}_{#1}}
\newcommand{\niamachs}[1]{\tilde{M}_{#1}}

\newcommand{\ajobs}[1]{J_{#1}}
\newcommand{\iajobs}[1]{\check{J}_{#1}}
\newcommand{\niajobs}[1]{\tilde{J}_{#1}}

\newcommand{\desc}{\mathrm{desc}}

\newcommand{\Opt}{\mathrm{OPT}}

\newcommand{\numload}{L}

\newcommand{\Oh}{\mathcal{O}}

\newcommand{\twp}{\mathrm{tw}_{\mathrm{p}}}
\newcommand{\twd}{\mathrm{tw}_{\mathrm{d}}}
\newcommand{\twi}{\mathrm{tw}_{\mathrm{i}}}

\newcommand{\er}{r}
\newcommand{\pe}{p}
\newcommand{\el}{\ell}
\newcommand{\eler}{\set{\el,\er}}

\newcommand{\sched}{S}
\newcommand{\niasched}{\tilde{S}}

\newcommand{\cutrank}{\mathrm{cutrk}}
\newcommand{\rankwidth}{\mathrm{rw}}
\newcommand{\cliquewidth}{\mathrm{cw}}

\newcommand{\cutjobs}[1]{J_{#1}}
\newcommand{\cutmachs}[1]{M_{#1}}

\newcommand{\sendjobs}[1]{\vec{J}_{#1}}
\newcommand{\splitjobs}[1]{\vec{J}_{#1}}

\newcommand{\typesizes}[1]{\varphi_{#1}}
\newcommand{\typemachs}[1]{M_{#1}}

\newcommand{\ZZ}{\mathbb{Z}}

\newcommand{\eps}{\varepsilon}

\DeclarePairedDelimiter\parens{(}{)}

\DeclarePairedDelimiter\set{\lbrace}{\rbrace}
\DeclarePairedDelimiterX\sett[2]{\lbrace}{\rbrace}{ #1 \,\delimsize| \,\mathopen{} #2 }

\theoremstyle{plain}

\title{Structural Parameters for Scheduling with Assignment Restrictions\footnote{This work was partially supported by the DAAD (Deutscher Akademischer Austauschdienst) and by the German Research Foundation (DFG) project JA 612/15-1.}}

\author[1]{Klaus Jansen}
\author[1]{Marten Maack}
\author[2]{Roberto Solis-Oba}
\affil[1]{University of Kiel, Kiel, Germany, \texttt{\{kj,mmaa\}@informatik.uni-kiel.de}}
\affil[2]{Western University, London, Canada, \texttt{solis@csd.uwo.ca}}

\begin{document}

\maketitle

\begin{abstract}
We consider scheduling on identical and unrelated parallel machines with job assignment restrictions. 
These problems are \textsf{NP}-hard and they do not admit polynomial time approximation algorithms with approximation ratios smaller than $1.5$ unless \textsf{P=NP}. 
However, if we impose limitations on the set of machines that can process a job, the problem sometimes becomes easier in the sense that algorithms with approximation ratios better than $1.5$ exist. 
We introduce three graphs, based on the assignment restrictions and study the computational complexity of the scheduling problem with respect to structural properties of these graphs, in particular their tree- and rankwidth. 
We identify cases that admit polynomial time approximation schemes or FPT algorithms, generalizing and extending previous results in this area.
\end{abstract}

\section{Introduction}\label{sec:intro}

We consider the problem of makespan minimization for scheduling on unrelated parallel machines.
In this problem a set $\jobs$ of $n$ jobs has to be assigned to a set $\machs$ of $m$ machines via a schedule $\sigma:\jobs\rightarrow\machs$.
A job $j$ has a processing time $p_{ij}$ for every machine $i$ and the goal is to minimize the makespan $\Cmax(\sigma)=\max_{i} \sum_{j\in\sigma^{-1}(i)}p_{ij}$.
In the three-field notation this problem is denoted by $\RCm$.
On some machines a job might have a very high, or even infinite processing time, so it should never be processed on these machines.
This amounts to assignment restrictions in which for every job $j$ there is a subset $\validmachs{j}$ of machines on which it may be processed.
An important special case of $\RCm$ is given if the machines are identical in the sense that each job $j$ has the same processing time $p_j$ on all the machines on which it may be processed, i.e., $p_{ij}\in\set{p_j,\infty}$.
This problem is sometimes called restricted assignment and is denoted as $\RA$ in the three-field notation.

We study versions of $\RCm$ and $\RA$ where the restrictions are in some sense well structured.
In particular we consider three different graphs that are defined based on the job assignment restrictions and study how structural properties of these graphs affect the computational complexity of the corresponding scheduling problems.
We briefly describe the graphs.
In the \emph{primal graph} the vertices are the jobs and two vertices are connected by an edge, iff there is a machine on which both of the jobs can be processed.
In the \emph{dual graph}, on the other hand, the machines are vertices and two of them are adjacent, iff there is a job that can be processed by both machines.
Lastly we consider the \emph{incidence graph}.
This is a bipartite graph and both the jobs and machines are vertices.
A job $j$ is adjacent to a machine $i$, if $i\in\validmachs{j}$. 
In Figure \ref{fig:graphs} an example of each graph is given.
These graphs have also been studied in the context of constraint satisfaction (see e.g. \cite{Sze03} or \cite{SS10}) and we adapted them for machine scheduling.

We consider the above scheduling problems in the contexts of parameterized and approximation algorithms.
For $\alpha>1$ an \emph{$\alpha$-approximation} for a minimization problem computes a solution of value $A(I)\leq\alpha\Opt(I)$, where $\Opt(I)$ is the optimal value for a given instance $I$.
A family of algorithms consisting of $(1+\eps)$-approximations for each $\eps>0$ with running times polynomial in the input length (and $1/\eps$) is called a \emph{(fully) polynomial time approximation scheme} (F)PTAS.
Let $\pi$ be some parameter defined for a given problem, and let $\pi(I)$ be its value for instance $I$.
The problem is said to be \emph{fixed-parameter tractable} (FPT) for $\pi$, if there is an algorithm that given $I$ and $\pi(I)=k$ solves $I$ in time $\Oh(f(k)|I|^c)$, where $c$ is a constant, $f$ any computable function and $|I|$ the input length.
This definition can easily be extended to multiple parameters.

\paragraph{Related work.}

\begin{wrapfigure}{r}{3.0cm} 
\centering
\includegraphics[scale=0.8]{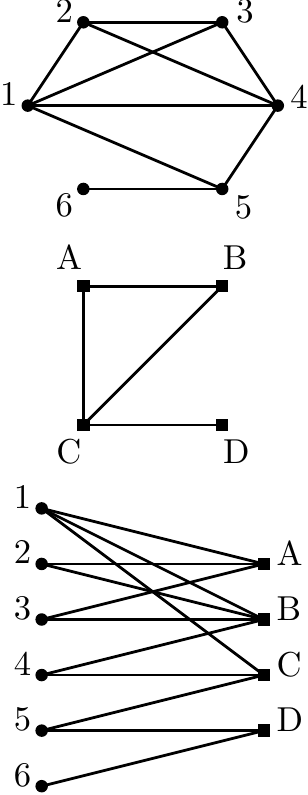}
\caption{Primal, dual and incidence graph for an instance with $6$ jobs and $4$ machines.}
\label{fig:graphs}
\end{wrapfigure}
In 1990 Lenstra, Shmoys and Tardos \cite{LST90} showed, in a seminal work, that there is a $2$-approximation for $\RCm$ and that the problem cannot be approximated with a ratio better than $1.5$ unless \textsf{P}$=$\textsf{NP}.
Both bounds also hold for $\RA$ and have not been substantially improved since that time.
The case where the number of machines is constant is weakly \textsf{NP}-hard and there is an FPTAS for this case \cite{HS76}.
In 2012 Svensson \cite{Sve12} presented an interesting result for $\RA$:
A special case of the restricted assignment problem called graph balancing was studied by Ebenlendr et al. \cite{EKS14}.
In this variant each job can be processed by at most $2$ machines and therefore an instance can be seen as a (multi-)graph where the machines are vertices and the jobs edges.
They presented a $1.75$ approximation for this problem and also showed that the $1.5$ inapproximability result remains true.
Lee et al. \cite{LLP09} studied the version of graph balancing where (in our notation) the dual graph is a tree and showed that there is an FPTAS for it.
Moreover, the special case of graph balancing where the graph is simple has been considered.
For this problem Asahiro et al. \cite{AMO11} presented among other things a pseudo-polynomial time algorithm for the case of graphs with bounded treewidth. 
For certain cases of $\RA$ with job assignment restrictions that are in some sense well-structured PTAS results are known.
In particular for the \emph{path- and tree-hierarchical} cases (\cite{OLL08} and \cite{EL11}) in which the machines can be arranged in a path or tree and the jobs can only be processed on subpaths starting at the leftmost machine or at the root machine respectively, and the \emph{nested} case (\cite{MSW10}), where $\validmachs{j}\subseteq\validmachs{j'}$, $\validmachs{j'}\subseteq\validmachs{j}$ or $\validmachs{j}\cap\validmachs{j'}=\emptyset$ holds for each pair of jobs $(j,j')$.

The study of $\RCm$ from the FPT perspective has started only recently.
Mnich and Wiese \cite{MW15} showed that $\RCm$ is FPT for the pair of parameters $m$ and the number of distinct processing times.
The problem is also FPT for the parameter pair $\max p_{ij}$ and the number of machine types \cite{KK16}.
Two machines have the same type, if each job has the same processing time on them. 
Furthermore Szeider \cite{Sze08} showed that graph balancing on simple graphs with unary encoding of the processing times is not FPT for the parameter treewidth under usual complexity assumptions.

\paragraph{Results.}

In this paper we present a graph theoretical viewpoint for the study of scheduling problems with job assignment restrictions that we believe to be of independent interest.
Using this approach we identify structural properties for which the problems admit approximation schemes or FPT algorithms, generalizing and extending previous results in this area. 
The results are based on dynamic programming utilizing tree and branch decompositions of the respective graphs.
For the approximation schemes the dynamic programs are combined with suitable rounding approaches.

Tree and branch decompositions are associated with certain structural \emph{width parameters}.
We consider two of them: treewidth and rankwidth.
In the following we denote the treewidth of the primal, dual and incidence graph with $\twp$, $\twd$ and $\twi$, respectively.
For the definitions of these concepts we refer to Section \ref{sec:prelim}.

We now describe our results in more detail.
Let $\validjobs{i}$ be the set of jobs the machine $i$ can process.
In the context of parameterized algorithms we show the following.
\begin{theorem}\label{thm:tw_primal_result_FPT}
$\RCm$ is FPT for the parameter $\twp$.
\end{theorem}
\begin{theorem}\label{thm:tw_dual_inst_result_FPT}
$\RCm$ is FPT for the pair of parameters $k_1,k_2$ with $k_1\in\set{\twd,\twi}$ and $k_2\in\set{\Opt,\max_i|\validjobs{i}|}$.
\end{theorem}
Note that $\RCm$ with constant $k_2$ remains NP-hard \cite{EKS14}.
In the context of approximation we get:
\begin{theorem}\label{thm:tw_dual_inst_result_APPROX}
$\RCm$ is weakly NP-hard, if $\twd$ or $\twi$ is constant and there is an FPTAS for both of these cases.
\end{theorem}
The hardness is due to the hardness of scheduling on two identical parallel machines $P2||\Cmax$.
The result for the dual graph is a generalization of the result in \cite{LLP09} and resolves cases that were marked as open in that paper.
All results mentioned so far are discussed in Section \ref{sec:tw}.
In the following section we consider the rankwidth:
\begin{theorem}\label{thm:rw_result}
There is a PTAS for instances of $\RA$ where the rankwidth of the incidence graph is bounded by a constant.
\end{theorem}
It can be shown that instances of $\RA$ with path- or tree-hierarchical or nested restrictions are special cases of the case when the incidence graph is a bicograph.
Bicographs are known to have a rankwidth of at most $4$ (see \cite{GV00}) and a suitable branch decomposition can be found very easily.
Therefore we generalize and unify the known PTAS results for $\RA$ with structured job assignment restrictions.


\section{Preliminaries}\label{sec:prelim}

In the following $I$ will always denote an instance of $\RCm$ or $\RA$ and most of the time we will assume that it is feasible.
We call an instance feasible if $\validmachs{j}\neq\emptyset$ for every job $j\in\jobs$.
A schedule is feasible if $\sigma(j)\in\validmachs{j}$.
For a subset $J\subseteq\jobs$ of jobs and a subset $M\subseteq\machs$ of machines we denote the subinstance of $I$ induced by $J$ and $M$ with $I[J,M]$.
Furthermore, for a set $S$ of schedules for $I$ we let $\Opt(S)=\min_{\sigma\in S}\Cmax(\sigma)$, and $\Opt(I)=\Opt(S)$ if $S$ is the set of all schedules for $I$.
We will sometimes use $\Opt(\emptyset)=\infty$.
Note that there are no schedules for instances without machines.
On the other hand, if $I$ is an instance without jobs, we consider the empty function a feasible schedule (with makespan $0$), and have therefore $\Opt(I)=0$ in that case.

\paragraph{Dynamic programs for $\RCm$.}

We sketch two basic dynamic programs that will be needed as subprocedures in the following.
The first one is based on iterating through the machines.
Let $\Opt(i,J)=\Opt(I[\jobs\setminus J,[i]])$ for $J\subseteq\jobs$ and $i\in [m]:=\set{1,\dots,m}$, assuming $\machs=[m]$.
Then it is easy to see that $
\Opt(i,J) = \min_{J \subseteq J' \subseteq \jobs}\max\set{\Opt(i-1, J'),\sum_{j\in J'\setminus J}p_{ij}}$.
Using this recurrence relation a simple dynamic program can be formulated that computes the values $\Opt(i,J)$.
It holds that $\Opt(I)=\Opt(m,\emptyset)$ and as usual for dynamic programs an optimal schedule can be recovered via backtracking.
The running time of such a program can be bounded by $2^{\Oh(n)}\times\Oh(m)$, yielding the following trivial result:
\begin{remark}
$\RCm$ is FPT for the parameter $n$.
\end{remark} 

The second dynamic program is based on iterating through the jobs.
Let $\lambda\in\ZZ_{\geq 0}^{\machs}$.
We call $\lambda$ a \emph{load vector} and say that a schedule $\sigma$ fulfils $\lambda$, if $\lambda_i=\sum_{j\in\sigma^{-1}(i)}p_{ij}$.
For $j\in [n]$ let $\Lambda(j)$ be the set of load vectors that are fulfilled by some schedule for the subinstance $I[[j],\machs]$, assuming $\jobs=[n]$.
Then $\Lambda(j)$ can also be defined recursively as the set of vectors $\lambda$ with $\lambda_{i^*}=\lambda'_{i^*}+p_{i^*j}$ and $\lambda_{i}=\lambda'_{i}$ for $i\neq i^*$, where $i^*\in\validmachs{j}$ and $\lambda'\in\Lambda(j-1)$.
Using this, a simple dynamic program can be formulated that computes $\Lambda(j)$ for all $j\in[n]$.
$\Opt(I)$ can be recovered from $\Lambda(n)$ and a corresponding schedule can be found via backtracking.
Let there be a bound $\numload$ for the number of distinct loads that can occur on each machine, i.e., $|\sett{\sum_{j\in\sigma^{-1}(i)}p_{ij}}{\sigma \text{ schedule for }I}|\leq\numload$ for each $i\in\machs$.
Then the running time can be bounded by $\numload^{\Oh(m)}\times\Oh(n)$, yielding:
\begin{remark}
$\RCm$ is FPT for the pair of parameters $m$ and $k$ with $k\in\set{\Opt,\max_i|\validjobs{i}|}$.
\end{remark}
For this note that both $\Opt$ and $2^{\max_i|\validjobs{i}|}$ are bounds for the number of distinct loads that can occur on any machine.
This dynamic program can also be used to get a simple FPTAS for $\RCm$ for the case when the number of machines $m$ is constant.
For this let $B$ be an upper bound of $\Opt(I)$ with $B\leq 2\Opt$.
Such a bound can be found with the $2$-approximation by Lenstra et al. \cite{LST90}.
Moreover let $\eps>0$ and $\delta=\eps/2$.
By rounding the processing time of every job up to the next integer multiple of $\delta B/n$ we get an instance $I'$ whose optimum makespan is at most $\eps \Opt(I)$ bigger than $\Opt(I)$.
The dynamic program can easily be modified to only consider load vectors for $I'$, where all loads are bounded by $(1+\delta/n)B$.
Therefore there can be at most $n/\delta+2$ distinct load values for any machine and an optimal schedule for $I'$ can be found in time $(n/\eps)^{\Oh(m)}\times\Oh(n)$.
The schedule can trivially be transformed into a schedule for the original instance without an increase in the makespan.

\paragraph{Tree decompostion and treewidth.}

A \emph{tree decomposition} of a graph $G$ is a pair $(T, \sett{X_t}{t\in V(T)} )$, where $T$ is a tree, $X_t\subseteq V(G)$ for each $t\in V(t)$ is a set of vertices of $G$, called a bag, and the following three conditions hold:
\begin{enumerate}
\item \label{enum:tree_decomp_union}$\bigcup_{t\in V(T)}X_t=V(G)$
\item \label{enum:tree_decomp_charedbag}$\forall \set{u,v}\in E(G)\exists t\in V(T): u,v\in X_t$
\item \label{enum:tree_decomp_connected}For every $u\in V(G)$ the set $T_u:=\sett{t\in V(T)}{u\in X_t}$ induces a connected subtree of $T$.
\end{enumerate}
The \emph{width} of the decomposition is $\max_{t\in V(T)}(|X_t|-1)$, and the \emph{treewidth} $\treewidth(G)$ of $G$ is the minimum width of all tree decompositions of $G$.
It is well known that forests are exactly the graphs with treewidth one, and that the treewidth of $G$ is at least as big as the biggest clique in $G$ minus $1$.
More precisely, for each set of vertices $V'\subseteq V(G)$ inducing a clique in $G$, there is a node $t\in V(T)$ with $V'\subseteq X_t$ (see e.g. \cite{Bod98}).
For a given graph and a value $k$ it can be decided in FPT time (and linear in $|V(G)|$) whether the treewidth of $G$ is at most $k$ and in the affirmative case a corresponding tree decomposition with $\Oh(k|V(G)|)$ nodes can be computed \cite{Bod96}.
However, deciding whether a graph has a treewidth of at most $k$, is \textsf{NP}-hard \cite{ACP87}.

\paragraph{Branch decomposition and rankwidth.}

It is easy to see that graphs with a small treewidth are sparse.
Probably the most studied parameter for \emph{dense} graphs is the cliquewidth $\cliquewidth(G)$.
In this paper however we are going to consider a related parameter called the rankwidth $\rankwidth(G)$.
These two parameters are equivalent in the sense that $\rankwidth(G)\leq\cliquewidth(G)\leq 2^{\rankwidth(G)+ 1}-1 $ \cite{OS06}.
Furthermore it is known that $\cliquewidth(G)\leq 3\times 2^{\treewidth(G)-1}$ \cite{CR05}.
On the other hand $\treewidth(G)$ cannot be bounded by any function in $\cliquewidth(G)$ or $\rankwidth(G)$, which can easily be seen by considering complete graphs.

A \emph{cut} of $G$ is a partition of $V(G)$ into two subsets.
For $X,Y\subseteq V(G)$ let $A_G[X,Y]=(a_{xy})$ be the $|X|\times |Y|$ adjacency submatrix induced by $X$ and $Y$, i.e., $a_{xy}=1$ if $\set{x,y}\in E(G)$ and $a_{xy}=0$ otherwise for $x\in X$ and $y\in Y$.
The \emph{cut rank} of $(X,Y)$ is the rank of $A_G[X,Y]$ over the field with two elements GF($2$) and denoted by $\cutrank_G(X,Y)$.
A \emph{branch decomposition} of $V(G)$ is a pair $(T,\eta)$, where $T$ is a tree with $|V(G)|$ leaves whose internal nodes have all degree $3$, and $\eta$ is a bijection from  $V(G)$ to the leafs of $T$.
For each $e=\set{s,t}\in E(T)$ there is an induced cut $\set{X_s,X_t}$ of $G$:
For $x\in\set{s,t}$ the set $X_x$ contains exactly the nodes $\eta^{-1}(\el)$, where $\el\in V(T)$ is a leaf that is in the same connected component of $T$ as $x$, if $e$ is removed.
Now the \emph{width} of $e$ (with respect to $\cutrank_G$) is $\cutrank_G(X_s,X_t)$ and the \emph{rankwidth} of the decomposition $(T,\eta)$ is the maximum width over all edges of $T$.
The \emph{rankwidth} of $G$ is the minimum rankwidth of all branch decompositions of $G$.
It is well known that the cliquewidth of a complete graph is equal to $1$ and this is also true for the rankwidth. 
For a given graph and fixed $k$ there is an algorithm that finds a branch decomposition of width $k$ in FPT-time (cubic in $|V(G)|$), or reports correctly that none exists \cite{HO08}.

\section{Treewidth Results}\label{sec:tw}

We start with some basic relationships between different restriction parameters for $\RCm$, especially the treewidths of the different graphs for a given instance.
Similar relationships have been determined for the three graphs in the context of constraint satisfaction.
\begin{remark}
$\twp\geq\max_i|\validjobs{i}|$ and $\twd\geq\max_j|\validmachs{j}|$.
\end{remark}
To see this note that the sets $\validjobs{i}$ and $\validmachs{j}$ are cliques in the primal and dual graphs, respectively.
\begin{remark}\label{rem:twpdi}
$\twi\leq\twp+1$ and $\twi\leq\twd+1$.
On the other hand $\twp\leq (\twi+1)\max_i |\validjobs{i}|-1$ and $\twd\leq (\twi+1)\max_j |\validmachs{j}|-1$.
\end{remark}
These properties were pointed out by Kalaitis and Vardi \cite{KV98} in a different context.
Note that this Remark together with Theorem \ref{thm:tw_primal_result_FPT} implies the results of Theorem \ref{thm:tw_dual_inst_result_FPT} concerning the parameter $\max_i|\validjobs{i}|$.
Furthermore, in the case of $\PCm$ with only $1$ job and $m$ machines, or $n$ jobs and only $1$ machine the primal graph has treewidth $0$ or $n-1$ and the dual $m-1$ or $0$, respectively, while the incidence graph in both cases  has treewidth $1$.

\subsection*{Dynamic Programs}\label{sec:tw_dyn}

We show how a tree decomposition $(T, \sett{X_t}{t\in V(T)} )$ of width $k$ for any one of the three graphs can be used to design a dynamic program for the corresponding instance $I$ of $\RCm$.
Selecting a node as the root of the decompostion, the dynamic program works in a bottom-up manner from the leaves to the root.
We assume that the decomposition has the following simple form:
For each leaf node $t\in V(T)$ the bag $X_t$ is empty and we fix one of these nodes as the root $a$ of $T$.
Furthermore each internal node $t$ has exactly two children $\el(t)$ and $\er(t)$ (left and right), and each node $t\neq a$ has one parent $\pe(t)$.
We denote the descendants of $t$ with $\desc(t)$.
A decomposition of this form can be generated from any other one without increasing the width and growing only linearly in size through the introduction of dummy nodes.
The bag of a dummy node is either empty or identical to the one of its parent.
 
For each of the graphs and each node $t\in V(T)$ we define sets $\iajobs{t}\subseteq\jobs$ and $\iamachs{t}\subseteq\machs$ of \emph{inactive} jobs and machines along with sets $\ajobs{t}$ and $\amachs{t}$ of \emph{active} jobs and machines.
The active jobs and machines in each case are defined based on the respective bag $X_t$, and the inactive ones have the property that they were active for a descendant $t\in\desc(t)$ of $t$ but are not at $t$.
In addition there are \emph{nearly inactive} jobs $\niajobs{t}$ and machines $\niamachs{t}$, which are the jobs and machines that are deactivated when going from $t$ to its parent $\pe(t)$ (for $t=a$ we assume them to be empty).
The sets are defined so that certain conditions hold.
The first two are that the (nearly) inactive jobs may only be processed on active or inactive machines, and the (nearly) inactive machines can only process active or inactive jobs:
\begin{align}
\validmachs{\iajobs{t}\cup\niajobs{t}} &\subseteq \amachs{t} \cup \iamachs{t}\label{eq:tw_nia_ia_jobs}\\
\validjobs{\iamachs{t}\cup\niamachs{t}} &\subseteq \ajobs{t} \cup \iajobs{t}\label{eq:tw_nia_ia_machs}
\end{align} 
Where $\validmachs{J^*}=\bigcup_{j\in J^*}\validmachs{j}$ and $\validjobs{M^*}=\bigcup_{i\in M^*}\validjobs{i}$ for any sets $J^*\subseteq\jobs$ and $M^*\subseteq\machs$.
Furthermore the (nearly) inactive jobs and machines of the children of an internal $t$ form a disjoint union of the inactive jobs and machines of $t$, respectively:
\begin{align}
\iajobs{t}  &= \iajobs{\el(t)}\ \dot{\cup}\ \niajobs{\el(t)}\ \dot{\cup}\ \iajobs{\er(t)}\ \dot{\cup}\ \niajobs{\er(t)}\label{eq:tw_disjoint_union_jobs}\\
\iamachs{t} &= \iamachs{\el(t)}\ \dot{\cup}\ \niamachs{\el(t)}\ \dot{\cup}\ \iamachs{\er(t)}\ \dot{\cup}\ \niamachs{\er(t)}\label{eq:tw_disjoint_union_machs}
\end{align}
Where $A\ \dot{\cup}\ B$ for any two sets $A,B$ emphasizes that the union $A\cup B$ is disjoint, i.e., $A\cap B=\emptyset$. 
Now at each node of the decomposition the basic idea is to perform three steps:
\begin{enumerate}
\item \label{enum:tw_step_combine1} Combine the information from the children (for internal nodes).
\item \label{enum:tw_step_sched} Consider the nearly inactive jobs and machines:
\begin{itemize}
\item \emph{Primal and incidence graph:} Try all possible ways of scheduling active jobs on nearly inactive machines.
\item \emph{Dual and incidence graph:} Try all possible ways of scheduling nearly inactive jobs on active machines.
\end{itemize}
\item \label{enum:tw_step_combine2} Combine the information from the last two steps.
\end{enumerate}
For the second step the dynamic programs described in Section \ref{sec:prelim} are used as subprocedures.
We now consider each of the three graphs.

\paragraph{The primal graph.}

In the primal graph all the vertices are jobs, and we define the active jobs of a tree node $t$ to be exactly the jobs that are included in the respective bag, i.e., $\ajobs{t}=X_t$.
The inactive jobs are those that are not included in $X_t$ but are in a bag of some descendant of $t$ and the nearly inactive one are those that are active at $t$ but inactive at $\pe(t)$, i.e., $\iajobs{t}=\sett{j\in\jobs}{j\not\in X_t \wedge \exists t'\in \desc(t): j\in X_{t'}}$ and $\niajobs{t}=\ajobs{t}\setminus\ajobs{\pe(t)}$.
Moreover the inactive machines are the ones on which some inactive job may be processed, and the (nearly in-)active machines are those that can process (nearly in-)active jobs and are not inactive, i.e., $\iamachs{t}=\validmachs{\iajobs{t}}$, $\amachs{t}=\validmachs{\ajobs{t}}\setminus\iamachs{t}$ and $\niamachs{t}=\validmachs{\niajobs{t}}\setminus\iamachs{t}$.
For these definitions we get:
\begin{lemma}\label{lem:tw_pg_conditions}
The conditions (\ref{eq:tw_nia_ia_jobs})-(\ref{eq:tw_disjoint_union_machs}) hold, as well as:
\begin{align}
\validjobs{\niamachs{t}} &\subseteq \ajobs{t}\label{eq:tw_pg_niamachs} \\
\validmachs{\iajobs{t}\cup\niajobs{t}} &= \niamachs{t}\cup\iamachs{t} \label{eq:tw_pg_niajobs}
\end{align}
\end{lemma}
\begin{proof}
(\ref{eq:tw_nia_ia_jobs}) and (\ref{eq:tw_pg_niajobs}): 
\[
\validmachs{\iajobs{t}\cup\niajobs{t}} = \validmachs{\iajobs{t}}\cup\validmachs{\niajobs{t}} = \iamachs{t} \cup (\validmachs{\niajobs{t}}\setminus\iamachs{t}) = \iamachs{t} \cup \niamachs{t}
\]
This yields (\ref{eq:tw_pg_niajobs}) and (\ref{eq:tw_pg_niajobs}) implies (\ref{eq:tw_nia_ia_jobs}).

(\ref{eq:tw_nia_ia_machs}) and (\ref{eq:tw_pg_niamachs}):
Let $i\in\iamachs{t}\cup\niamachs{t}$ and $j\in\validjobs{i}$.
We first consider the case that $i\in\iamachs{t}$. 
Then there is a job $j'\in\iajobs{t}$ with $i\in\validmachs{j'}$.
If $j=j'$, we have $j\in\iajobs{t}$ and otherwise $\set{j,j'}\in E(G)$.
Because of (T2) there is a node $t'\in V(T)$ with $j,j'\in\ajobs{t'}$.
Since $j'\in\iajobs{t}$, we have $j'\not\in\ajobs{t}$.
This together with (T3) gives $t'\in\desc(t)$.
Now $j\not\in\ajobs{t}$ implies $j\in\iajobs{t}$.
Therefore we have $j\in\ajobs{t}\cup\iajobs{t}$.
Next we consider he case that $i\in\niamachs{t}$. 
In this case there is a job $j'\in\niajobs{t}$ with $i\in\validmachs{j'}$ and for each job $j''\in\iajobs{t}$ we have $i\not\in\validmachs{j''}$.
If $j=j'$ we have $j\in\ajobs{t}$ and otherwise $\set{j,j'}\in E(G)$.
Because of (T2) there is again a node $t'\in V(T)$ with $j,j'\in\ajobs{t'}$.
Since $j,j'\not\in\iajobs{t}$, $j'\not\in\ajobs{\pe(t)}$ and $j'\in\ajobs{t}$ we get $j\in\ajobs{t}$ using (T3).
This also implies (\ref{eq:tw_pg_niamachs}).

(\ref{eq:tw_disjoint_union_jobs}):
All but $(\iajobs{\el(t)}\cup\niajobs{\el(t)})\cap(\iajobs{\er(t)}\cup\niajobs{\er(t)})=\emptyset$ follows directly from the definitions.
Assuming there is a job $j\in(\iajobs{\el(t)}\cup\niajobs{\el(t)})\cap(\iajobs{\er(t)}\cup\niajobs{\er(t)})$ we get $j\in\ajobs{t}$ because of (T3), yielding a contradiction.

(\ref{eq:tw_disjoint_union_machs}):
Because of (\ref{eq:tw_disjoint_union_jobs}) and the definitions we get $\iamachs{t}=\iamachs{\el(t)}\cup\niamachs{\el(t)}\cup\iamachs{\er(t)}\cup\niamachs{\er(t)}$, and $\iamachs{s(t)}\cap\niamachs{s(t)}$ for $s\in\eler$ is clear by definition.
Therefore it remains to show $(\iamachs{\el(t)}\cup\niamachs{\el(t)})\cap(\iamachs{\er(t)}\cup\niamachs{\er(t)})=\emptyset$.
We assume that there is a machine $i$ in this cut.
Then there are jobs $j_s\in\iajobs{s(t)}\cup\niajobs{s(t)}$ for $s\in\eler$ with $i\in\validmachs{j_s}$.
We have $\set{j_\el,j_\er}\in E$ and because of (T2) there is a node $t'$ with $j_\el,j_\er\in\ajobs{t'}$.
Because of $j_\el,j_\er\not\in\ajobs{t}$ and (T3) we have a contradiction.
\end{proof}

For $J\subseteq\jobs$ and $M\subseteq\machs$ let $\Gamma(J,M)=\sett{J'\subseteq J}{\forall j\in J':\validmachs{j}\cap M\neq\emptyset}$.
Let $t\in V(T)$, $J\in\Gamma(\ajobs{t},\iamachs{t})$ and $J'\in\Gamma(\ajobs{t}\setminus\niajobs{t},\iamachs{t}\cup\niamachs{t})$.
We set $\sched(t,J)$ and $\niasched(t,J')$ to be the sets of feasible schedules for the instances $I[\iajobs{t}\cup J,\iamachs{t}]$ and $I[\iajobs{t}\cup\niajobs{t}\cup J', \iamachs{t}\cup\niamachs{t}]$ respectively.
We will consider $\Opt(\sched(t,J))$ and $\Opt(\niasched(t,J'))$.

First note that $\Opt(I)=\Opt(\sched(a,\emptyset))$, where $a$ is the root of $T$.
Moreover, for a leaf node $t$ there are neither jobs nor machines and $\Opt(\sched(t,\emptyset))=\Opt(\niasched(t,\emptyset))=\Opt(\set{\emptyset})=0$ holds.
Hence let $t$ be a non-leaf node.
We first consider how $\Opt(\sched(t,J))$ can be computed from the children of $t$ (Step \ref{enum:tw_step_combine1}).
Due to Property \ref{enum:tree_decomp_connected} of the tree decomposition and (\ref{eq:tw_nia_ia_jobs}) the jobs from $J$ are already active on at least one of the direct descendants of $t$.
Because of this and (\ref{eq:tw_disjoint_union_machs}), $J$ may be split in two parts $J_{\el}\dot{\cup}J_{\er}=J$, where $J_s\in\Gamma(\ajobs{s(t)}\setminus\niajobs{s(t)},\iamachs{s(t)}\cup\niamachs{s(t)})$ for $s\in\eler$.
Let $\Phi(J)$ be the set of such pairs $(J_{\el},J_{\er})$.
From (\ref{eq:tw_disjoint_union_jobs}), (\ref{eq:tw_disjoint_union_machs}) and (\ref{eq:tw_pg_niajobs}) we get:
\begin{lemma}\label{lem:tw_pg_reccurence1}
$\Opt\parens{\sched(t,J)} = \min_{(J_{\el},J_{\er})\in \Phi(J)}\max_{s\in\eler} \Opt(\niasched(s(t),J_s))$.
\end{lemma}
\begin{proof}
Let $\sigma^*\in\sched(t,J)$ be optimal.
Since $J\subseteq\ajobs{t}$, we have $J\cap\niajobs{s(t)}=\emptyset$ for $s\in\eler$.
Let $J^*_s=\sigma'^{*-1}(\iamachs{s(t)}\cup\niamachs{s(t)})\cap J$.
Because of (\ref{eq:tw_disjoint_union_machs}) we have $J=J^*_\el \dot{\cup} J^*_\er$ and $J^*_s\in\Gamma(\ajobs{s(t)}\setminus\niajobs{s(t)},\iamachs{s(t)}\cup\niamachs{s(t)})$ obviously holds.
Let $\sigma^*_s=\sigma^*|_{J'_s\cup\niajobs{s(t)}}$.
Because of (\ref{eq:tw_pg_niajobs}) we have $\sigma^*_s\in\niasched(s(t),J^*_s)$ and (\ref{eq:tw_disjoint_union_jobs}) implies $\sigma^* = \sigma^*_\el\cup\sigma^*_\er$.
This yields:
\begin{align*}
\Opt(\sched(t,J)) 	&= \Cmax(\sigma^*)\\
					&= \max_{s\in\eler}\Cmax(\sigma^*_s)\\
					&\geq \max_{s\in\eler} \Opt(\niasched(s(t),J^*_s))\\
					&\geq \min_{(J_{\el},J_{\er})\in \Phi(J)}\max_{s\in\eler} \Opt(\niasched(s(t),J_s))
\end{align*}
Now let $(J_{\el},J_{\er})\in \Phi(J)$ minimizing the righthand side of the equation and $\sigma_s\in\niasched(s(t),J_s)$ optimal.
Then (\ref{eq:tw_disjoint_union_jobs}) and (\ref{eq:tw_disjoint_union_machs}) imply that $\sigma:=\sigma_\el\cup\sigma_\er$ is in $S(t,J)$.
Therefore we have $\Cmax(\sigma)\geq\Cmax(\sigma^*)$.
Since $\Cmax(\sigma)$ also equals the right hand side of the equation the claim follows.
\end{proof}
Consider the computation of $\Opt(\niasched(t,J'))$ (Step \ref{enum:tw_step_combine2}).
We may split $J'$ and $\niajobs{t}$ into a set going to the nearly inactive and a set going to the inactive machines.
We set $\Psi(J')$ to be the set of pairs $(A,X)$ with $J'\cup\niajobs{t}=A\dot{\cup} X$, $A\in\Gamma(\niajobs{t}\cup J',\niamachs{t})$ and $X\in\Gamma(\niajobs{t}\cup J',\iamachs{t})$.
Because of (\ref{eq:tw_disjoint_union_jobs})-(\ref{eq:tw_pg_niamachs}) we have:
\begin{lemma}\label{lem:tw_pg_reccurence2}
$\Opt\parens{ \niasched(t,J') } =  \min_{(A,X)\in\Psi(J')} \max\set{ \Opt\parens{\sched(t,X)}  , \Opt\parens{I[A,\niamachs{t}]} }$.
\end{lemma}
\begin{proof}
Let $\sigma^{*}\in\niasched(t,J')$ be optimal.
Because of (\ref{eq:tw_pg_niamachs}) we have $\sigma^{*-1}(\niamachs{t})\subseteq J'\cup\niajobs{t}$.
We set $A^*=\sigma^{*-1}(\niamachs{t})$ and $X^*=(J'\cup\niajobs{t})\setminus A$.
Then $(A^*,X^*)\in\Psi(J')$.
Let $\check{\sigma}^*=\sigma^*|_{\iajobs{t}\cup X^*}$ and $\tilde{\sigma}^*=\sigma^*|_{A}$.
Then $\check{\sigma}^*\in\sched(t,X^*)$ and $\tilde{\sigma}^*$ is a feasible schedule for $I[A^*,\niamachs{t}]$.
Because of (\ref{eq:tw_disjoint_union_jobs}) and (\ref{eq:tw_disjoint_union_machs}), we have $\sigma=\check{\sigma}^*\ \dot{\cup}\ \tilde{\sigma}^*$ and:
\begin{align*}
\Opt(\niasched(t,J'))	&= \Cmax(\sigma^*)\\
						&= \max\set{ \Cmax(\check{\sigma}^*), \Cmax(\tilde{\sigma}^*)}	\\
						&\geq \max\set{\Opt\parens{\sched(t,X^*)}  , \Opt\parens{I[A^*,\niamachs{t}]}}\\
						&\geq \min_{(A,X)\in\Psi(J')} \max\set{ \Opt\parens{\sched(t,X)}  , \Opt\parens{I[A,\niamachs{t}]} }
\end{align*}
Now let $(A,X)\in\Psi(J')$ minimizing the right hand side of the equation, $\check{\sigma}\in\sched(t,X)$ and $\tilde{\sigma}$ a feasible schedule for $I[A,\niamachs{t}]$.
Then (\ref{eq:tw_disjoint_union_jobs}) and (\ref{eq:tw_disjoint_union_machs}) yield $\sigma:=\check{\sigma}\ \dot{\cup}\  \tilde{\sigma}\in\niasched(t,J')$, and therefore $\Cmax(\sigma)\geq \Cmax(\sigma^*)$.
Since $\Cmax(\sigma)$ also equals the right hand side of the equation, the claim follows.
\end{proof}
Determining the values $\Opt\parens{I[A,\niamachs{t}]}$ corresponds to Step \ref{enum:tw_step_sched}.
Note that these values can be computed using the first dynamic program from Section \ref{sec:prelim} in time $2^{\Oh(k)}\times\Oh(m)$.

\paragraph{The dual graph.}

For the dual graph the (in-)active jobs and machines are defined dually:
The active machines for a tree node $t$ are the ones in the respective bag, the inactive machines are those that were active for some descendant but are not active for $t$, and the nearly inactive machines are those that are active at $t$ but inactive at its parent, i.e., $\amachs{t}=X_t$, $\iamachs{t}=\sett{i\in\machs}{i\not\in\amachs{t}\wedge \exists t'\in\desc(t):i\in X_{t'}}$ and $\niamachs{t}=\amachs{t}\setminus\iamachs{p(t)}$.
Furthermore the inactive jobs are those that may be processed on some inactive machine and the (nearly in-)active ones are those that can be processed on some (nearly in-)active machine and are not inactive, i.e., $\iajobs{t}=\validjobs{\iamachs{t}}$, $\ajobs{t}=\validjobs{\amachs{t}}\setminus\iajobs{t}$ and $\niajobs{t}=\validjobs{\niamachs{t}}\setminus\iajobs{t}$.
With these definitions we get analogously to Lemma \ref{lem:tw_pg_conditions}: 
\begin{lemma}
The conditions (\ref{eq:tw_nia_ia_jobs})-(\ref{eq:tw_disjoint_union_machs}) hold, as well as:
\begin{align}
\validmachs{\niajobs{t}} &\subseteq \amachs{t} \label{eq:tw_dg_niajobs}\\
\validjobs{\iamachs{t}\cup\niamachs{t}} &= \niajobs{t} \cup \iajobs{t} \label{eq:tw_dg_niamachs}
\end{align}
\qed
\end{lemma}

We will need some extra notation.
Like we did in Section \ref{sec:prelim} we will consider load vectors $\lambda\in\ZZ_{\geq 0}^{M}$, where $M\subseteq\machs$ is a set of machines.
We say that a schedule $\sigma$ fulfils $\lambda$, if $\lambda_i=\sum_{j\in\sigma^{-1}(i)}p_{ij}$ for each $i\in M$.
For any set $S$ of schedules for $I$ we denote the set of load vectors for $M$ that are fulfilled by at least one schedule from $S$ with $\Lambda(S,M)$.
Furthermore we denote the set of all schedules for $I$ with $S(I)$, and for a subset of jobs $J\subseteq\jobs$, we write $\Lambda(J,M)$ as a shortcut for $\Lambda(S(I[J,M]),M)$.
Let $t\in V(T)$.
We set $\sched(t)=\sched(I[\iajobs{t},\iamachs{t}\cup\amachs{t}])$ and $\niasched(t)=\sched(I[\iajobs{t}\cup\niajobs{t},\iamachs{t}\cup\amachs{t}])$. 
Moreover, for $\lambda\in \Lambda(\sched(t),\amachs{t})$ and $\lambda'\in \Lambda(\niasched(t),\amachs{t})$ we set $\sched(t,\lambda)\subseteq\sched(t)$ and $\niasched(t,\lambda')\subseteq\niasched(t)$ to be those schedules that fulfil $\lambda$ and $\lambda'$ respectively.
We now consider $\Opt(\sched(t,\lambda))$ and $\Opt(\niasched(t,\lambda'))$.

First note $\Opt(I)=\Opt(S(a,\emptyset))$.
Moreover, for a leaf node $t$ we have neither jobs nor machines and $\Lambda(\sched(t),\amachs{t})=\Lambda(\niasched(t),\amachs{t})=\set{\emptyset}$.
Therefore $\Opt(\sched(t,\emptyset))=\Opt(\niasched(t,\emptyset)) =\Opt(\set{\emptyset})=0$.
Hence let $t$ be a non-leaf node.
Again, we first consider how $\Opt(\sched(t,\lambda))$ can be computed from the children of $t$.
Because of (\ref{eq:tw_disjoint_union_jobs}) $\lambda$ may be split into a left and a right part.
For two machine sets $M,M'$ let $\tau_{M,M'}:\ZZ_{\geq 0}^{M}\rightarrow \ZZ_{\geq 0}^{M'}$ be a trasformation function for load vectors, where the $i$-th entry of $\tau_{M,M'}(\lambda)$ equals $\lambda_i$ for $i\in M\cap M'$ and $0$ otherwise.
We set $\Xi(\lambda)$ to be the set of pairs $(\lambda_{\el},\lambda_{\er})$ with $\lambda=\tau_{\amachs{\el(t)},\amachs{t}}(\lambda_{\el}) + \tau_{\amachs{\er(t)},\amachs{t}}(\lambda_{\er})$, and $\lambda_s\in\Lambda(\niasched(s(t)),\amachs{s(t)})$ for $s\in\eler$.
Because of (\ref{eq:tw_nia_ia_jobs}), (\ref{eq:tw_disjoint_union_jobs}) and (\ref{eq:tw_disjoint_union_machs}), we have:
\begin{lemma}\label{lem:tw_dg_reccurence1}
$\Opt(\sched(t,\lambda))= \min_{(\lambda_\el,\lambda_\er)\in\Xi(\lambda)}\max_{s\in\eler} \Opt(\niasched(s(t),\lambda_s))$.
\end{lemma}
\begin{proof}
Let $\sigma^*\in\sched(t,\lambda)$ be optimal.
Because of (\ref{eq:tw_nia_ia_jobs}) we have $\sigma^*(\iajobs{s(t)}\cup\niajobs{s(t)})\subseteq\amachs{s(t)}\cup\iamachs{s(t)}$ for $s\in\eler$.
Let $\sigma^*_s=\sigma^*|_{\iajobs{s(t)}\cup\niajobs{s(t)}}$ and $\lambda^*_s$ the load vector that $\sigma^*_s$ fulfils on $\amachs{s(t)}$.
Then we have $\sigma^*_s\in\niasched(s(t),\lambda^*_s)$ and $(\lambda^*_\el,\lambda^*_\er)\in\Xi(\lambda)$.
Because of (\ref{eq:tw_disjoint_union_jobs}) and (\ref{eq:tw_disjoint_union_machs}) we have $\sigma^*=\sigma^*_\el\ \dot{\cup}\ \sigma^*_\er$.
Because of the objective function we have:
\begin{align*}
\Opt(\sched(t,\lambda))	&= \Cmax(\sigma^*)\\
							&= \max_{s\in\eler} \Cmax(\sigma^*_s)\\
							&\geq \max_{s\in\eler} \Opt(\niasched(s(t),\lambda^*_s))\\
							&\min_{(\lambda_\el,\lambda_\er)\in\Xi(\lambda)}\max_{s\in\eler} \Opt(\niasched(s(t),\lambda_s))			
\end{align*}
Now let $(\lambda_\el,\lambda_\er)\in\Xi(\lambda)$ minimizing the right hand side of the equation and $\sigma_s\in\niasched(s(t),\lambda_s)$ optimal.
Then $\sigma:=\sigma_{\el}\cup\sigma_{\er}$ is in $\sched(t,\lambda)$ and $\Cmax(\sigma)$ equals the right hand side of the equation.
Since furthermore $\Cmax(\sigma)\geq\Cmax(\sigma^*)$ the claim follows.
\end{proof}

Now we consider $\Opt(\niasched(t,\lambda'))$.
We may split $\lambda'$ into the load due to inactive and that due to nearly inactive jobs.
Note that the nearly inactive jobs can only be processed by active machines (\ref{eq:tw_dg_niajobs}).
We set $\Upsilon(\lambda')$ to be the set of pairs $(\alpha,\xi)$ with $\lambda'=\alpha + \xi$, $\alpha\in\Lambda(\niajobs{t},\amachs{t})$ and $\xi\in\Lambda(\sched(t),\amachs{t})$.
Now (\ref{eq:tw_disjoint_union_jobs}), (\ref{eq:tw_disjoint_union_machs}) and (\ref{eq:tw_dg_niajobs}) yield:
\begin{lemma}\label{lem:tw_dg_reccurence2}
$\Opt(\niasched(t,\lambda'))=\min_{(\alpha,\xi)\in\Upsilon(\lambda')} \max\parens{\set{\Opt(\sched(t,\xi))} \cup  \sett{\lambda'_i}{i\in\amachs{t}}}$.
\end{lemma}
\begin{proof}
Let $\sigma^*\in\niasched(t,\lambda')$ be optimal.
Then (\ref{eq:tw_dg_niajobs}) implies $\sigma^*(\niajobs{t})\subseteq \amachs{t}$.
We set $\tilde{\sigma}^*=\sigma^*|_{\niajobs{t}}$ and $\check{\sigma}^*=\sigma^*|_{\iajobs{t}}$.
Furthermore let $\alpha^*$ be the load vector fulfilled by $\tilde{\sigma}^*$ and $\xi^*$ the one fulfilled by $\check{\sigma}^*$ on $\amachs{t}$. 
Then $\tilde{\sigma}^*$ is a feasible schedule for $I[\niajobs{t},\amachs{t}]$ fulfilling $\alpha^*$, $\check{\sigma}^*\in\sched(t,\xi^*)$ and $(\alpha^*,\xi^*)\in\Upsilon(\lambda')$.
Furthermore (\ref{eq:tw_disjoint_union_jobs}) yields $\sigma^*=\tilde{\sigma}^*\ \dot{\cup}\ \check{\sigma}^*$.
We get:
\begin{align*}
\Opt(\niasched(t,\lambda'))	&= \Cmax(\sigma^*)\\
								&= \max(\set{\Cmax(\check{\sigma}^*)}\cup\sett{\lambda'_i}{i\in\amachs{t}})\\
								&\geq \max(\set{\Opt(\sched(t,\xi^*))}\cup\sett{\lambda'_i}{i\in\amachs{t}})\\
								&\geq \min_{(\alpha,\xi)\in\Upsilon(\lambda')} \max(\set{\Opt(\sched(t,\xi))} \cup  \sett{\lambda'_i}{i\in\amachs{t}})
\end{align*}
Now let $(\alpha,\xi)\in\Upsilon(\lambda')$ minimizing the right hand side of the equation, $\check{\sigma}\in\sched(t,\xi)$ and $\tilde{\sigma}$ a feasible schedule for $I[\niajobs{t},\amachs{t}]$ fulfilling $\alpha$.
Then $\sigma:=\check{\sigma}\cup\tilde{\sigma}\in\niasched(t,\lambda')$ and therefore $\Cmax(\sigma)\geq\Cmax(\sigma^*)$. 
Since $\Cmax(\sigma)$ also equals the right hand side of the equation, the claim follows.
\end{proof}
The set $ \Lambda(\niajobs{t},\amachs{t})$ can be computed using the second dynamic program described in Section \ref{sec:prelim} in time $\numload^{\Oh(k)}\times \Oh(n)$ if $\numload$ is again a bound on the number of distinct loads that can occur on each machine.

\paragraph{The incidence graph.}

For the incidence graph we combine the ideas that we used for the two other graphs.
The situation is slightly more complicated because we have to handle the jobs and machines simultaneously.
All the job sets are defined like in the primal, and all the machine sets like in the dual graph case.
With these definitions the conditions (\ref{eq:tw_nia_ia_jobs})-(\ref{eq:tw_disjoint_union_machs}) follow almost directly from the definitions together with (T2) and (T3).
The proofs for the recurrence relations in this paragraph have the same structure as the proofs for the other recurrence relations and no new ideas are needed.
Therefore they are omitted.

Let $t\in V(t)$, $J\in\Gamma(\ajobs{t},\iamachs{t})$ and $J'\in\Gamma(\ajobs{t}\setminus\niajobs{t},\iamachs{t}\cup\niamachs{t})$.
We set $\sched(t,J)$ to be the set of feasible schedules $\sigma$ for $I[\iajobs{t}\cup J,\iamachs{t}\cup\amachs{t}]$ that schedule the jobs from $J$ on inactive machines, i.e., $\sigma(j)\in\iamachs{t}$ for each $j\in J$.
Moreover, $\niasched(t,J')$ is the set of schedules for $I[\iajobs{t}\cup\niajobs{t}\cup J',\iamachs{t}\cup\amachs{t}]$ that schedule the jobs from $J'$ on (nearly) inactive machines $\niamachs{t}\cup\iamachs{t}$.
The sets of schedules that in addition fulfil a load vector $\lambda\in\Lambda(\sched(t,J),\amachs{t})$ or $\lambda'\in\Lambda(\niasched(t,J')$ are denoted by $\sched(t,J,\lambda)$ and $\niasched(t,J',\lambda')$.
We consider $\Opt(\sched(t,J,\lambda))$ and $\Opt(\niasched(t,J',\lambda'))$.

First note $\Opt(I)=\Opt(S(a,\emptyset,\emptyset))$.
For a leaf note $t$ there are neither jobs nor machines and therefore $\Opt(\sched(t,\emptyset,\emptyset))=\Opt(\emptyset)) = 0$.
Hence let $t$ be a non-leaf node.
Like before, we first consider $\Opt(\sched(t,J,\lambda))$.
Both $J$ and $\lambda$ may be split into a left and a right part and we set $\Phi(J)$ like before.
Moreover, for $(J_{\el},J_{\er})\in\Phi(J)$ we define $\Xi(\lambda,(J_{\el},J_{\er}))$ to be the set of pairs $(\lambda_{\el},\lambda_{\er})$ with $\lambda_s\in\Lambda(\niasched(s(t),J_s),\amachs{s(t)})$ for $s\in\eler$.
Due to (\ref{eq:tw_nia_ia_jobs})-(\ref{eq:tw_disjoint_union_machs}) we have:
\begin{lemma}\label{lem:tw_ig_reccurence1}
$\Opt(\sched(t,J,\lambda)) = \min_{(J_{\el},J_{\er}),(\lambda_{\el},\lambda_{\er})}\max_{s\in\eler}\Opt(\niasched(s(t),J_s,\lambda_s))$.
\qed
\end{lemma}

Next we consider $\Opt(\niasched(t,J',\lambda'))$.
The set $J'$ again may be split into a part going to the inactive and a part going to the nearly inactive machines, while the nearly inactive jobs $\niajobs{t}$ have to be split into a part going to the inactive and a part going to the active machines (note that in this case (\ref{eq:tw_dg_niajobs}) does not hold).
Therefore, we set $\Psi(J')$ to be the set of pairs $(A,X)$ with $J' = A\dot{\cup} X$, $A\cap J'\in\Gamma(J',\niamachs{t})$, $A\cap\niajobs{t}\in\Gamma(\niajobs{t},\amachs{t})$ and $X\in\Gamma(\niajobs{t}\cup J',\iamachs{t})$.
The splitting of $\lambda'$ is more complicated as well, because in this case all of the active machines may receive load from the nearly inactive jobs, and the  nearly inactive machines may additionally receive load from the active but not nearly inactive jobs ((\ref{eq:tw_pg_niamachs}) does not hold).
Therefore we set $\Upsilon(\lambda',(A,X))$ to be the set of triplets $(\alpha,\beta,\xi)$ with $\alpha\in\Lambda(A\cap J',\niamachs{t})$, $\beta\in\Lambda(A\cap\niajobs{t},\amachs{t})$, $\xi\in\Lambda(\sched(t,X),\amachs{t})$ and $\lambda'=\tau_{\niamachs{t},\amachs{t}}(\alpha) + \beta + \xi$.
Due to (\ref{eq:tw_disjoint_union_jobs}) and (\ref{eq:tw_disjoint_union_machs}) we have:
\begin{lemma}\label{lem:tw_ig_reccurence2}
$\Opt(\niasched(t,J',\lambda')) = \min_{(A,X),(\alpha,\beta,\xi)}\max\parens{\set{\Opt(\sched(t,X,\xi))}\cup\sett{\lambda'(i)}{i\in\amachs{t}}}$.
\qed
\end{lemma}
Note that the sets $\Lambda(A\cap J',\niamachs{t})$ and $\Lambda(A\cap\niajobs{t},\amachs{t})$ can be computed in time $\numload^{\Oh(k)}$ using the second dynamic program described in Section \ref{sec:prelim}, if $\numload$ is again a bound on the number of distinct loads that can occur on each machine.

\paragraph{Results.}

Using above arguments, we can design dynamic programs with running time $2^{\Oh(k)}\times\Oh(nm)$ in the primal case and $\numload^{\Oh(k)}\times\Oh(nm)$ in the dual and incidence graph cases.
Optimal schedules can be found via backtracking proving the Theorems \ref{thm:tw_primal_result_FPT} and \ref{thm:tw_dual_inst_result_FPT}.
Theorem \ref{thm:tw_dual_inst_result_APPROX} follows by the combination of the dynamic programs and a rounding scheme similar to that in Section \ref{sec:prelim}.

\section{Rankwidth Results}\label{sec:rw}

First we want to argue that there is not much to be gained when considering primal or dual graphs with bounded rankwidth.
For this consider any instance $I$ of $\RA$.
By adding a job with processing time $\Opt(I)$ that can be processed on every machine, and a machine that can only process this new job, we get a modified instance $I'$.
Any schedule for one of the instances can trivially be transformed into a schedule for the other without an increase in the makespan.
However, while the rankwidth of the primal or dual graph of $I$ could have been arbitrarily high, the rankwidth of the primal and dual graph of $I'$ are both equal to one, because these graphs are complete.

We study the case when the rankwidth of the incidence graph is bounded by a constant $k$.
Moreover we assume that also the number $d$ of distinct job sizes is bounded by a constant, which we can do because of the following result.
Let $\mathcal{I}$ be some class of instances of $\RA$, which is invariant with respect to changing the processing times of jobs and the introduction of copies of jobs.
\begin{lemma}[Rounding Lemma]\label{lem:rw_rounding}
If there is a PTAS for instances from $\mathcal{I}$, for which the number of distinct processing times is bounded by a constant, then there is also a PTAS for any instance from $\mathcal{I}$.  
\end{lemma}
\begin{proof}
Let $I\in\mathcal{I}$, $\eps>0$ and $B$ an upper bound of $\Opt(I)$ with $B\leq 2\Opt$.
Such a bound $B$ can be found in polynomial time for example with the $2$-approximation by Lenstra et al. \cite{LST90}.
Moreover let $\delta:=\min\set{1/3,\eps/7}$.
We call jobs $j$ \emph{big}, if $p_j>\delta B$ and otherwise \emph{small}.
Next, we construct a modified instance $I'$.
This instance has the same machine set and for each big job $j$ a job with the same restrictions and processing time $p'_j:=\delta^2 B \lceil\frac{p_j}{\delta^2 B}\rceil$ is included in the job set.
This yields $p'_j\leq p_j+\delta^2 B\leq (1+\delta)p_j$.
For each small job $j$ in $I$ we introduce $\lceil\frac{np_j}{\delta B}\rceil\in\Oh(n)$ many jobs with the same restrictions as $j$, and with processing time $\frac{\delta B}{n}$.

Note that $I'$ has a has at most $1/\delta +1$ many distinct processing times and that $I'\in\mathcal{I}$.
Moreover the size of $I'$ is polynomial in the size of $I$.

Given an optimal solution of $I$, consider the solution of $I'$ we get by scheduling both the big and the small jobs in $I'$ the same way as there analogues in $I$.
The big jobs on a machine can cause an increase of the processing time of at most factor $(1+\delta)$, while for each small job of $I$ there may be an increase of at most $\frac{\delta B}{n}$.
Therefore we get: 
\begin{equation*}
\Opt(I')\leq \Opt(I)+\delta\Opt(I)+\delta B\leq (1+3\delta)\Opt(I)
\end{equation*}

Now given a PTAS for instances of $\mathcal{I}$ for which the number of distinct processing times is bounded by a constant, we can compute a schedule $\sigma'$ for $I'$ with $\Cmax(\sigma)\leq (1+\delta)\Opt(I')$ in polynomial time.
We use $\sigma'$ to construct a schedule for $\sigma$ for $I$.
In this schedule the big jobs are assigned in the same way as there analogous in $I'$.
For the small jobs we need some additional consideration.
Let $S$ and $S'$ be the set of small jobs in $I$ and $I'$ respectively.
Moreover, for $j\in S$ let $S'(j)$ be the set of small jobs that were inserted in $I'$ due to $j$.
The assignment of $S(j)$ in $\sigma'$ can be seen as a \emph{fractional} assignment of $j$.
We find a rounding for this fractional assignment of the small jobs.
For each machine $i$ and small job $j\in S$ let $x_{ij}=|\sett{j'\in S(j)}{\sigma'(j)=i}|/|S(j)|$.
Furthermore let $t_i$ be the summed up processing time that machine $i$ receives in the schedule $\sigma'$ from small jobs, i.e., $t_i=|\sett{j'\in S'}{\sigma'(j)=i}|\frac{\delta B}{n}$.
Then $(x_{ij})$ is a solution of the following linear program:
\begin{align}
\sum_{i\validmachs{j}}x_{ij} 	& = 1 		& \forall j\in S \label{eq:rw_rounding_lp1} \\
\sum_{j\in S}p_jx_{ij} 		& \leq t_i 	& \forall i\in\machs \textbf{\label{eq:rw_rounding_lp2}}\\
x_{ij}							& \geq 1	& \forall j\in S,i\in\machs \nonumber
\end{align}
Using the rounding approach by Lenstra et al. \cite{LST90}, we can transform this in polynomial time into an integral solution fulfilling the constraint (\ref{eq:rw_rounding_lp1}) and instead of (\ref{eq:rw_rounding_lp2}) the modified constraint:
\begin{align*}
\sum_{j\in S}p_jx_{ij} 		& \leq t_i+\max_{j\in S} p_j 	& \forall i\in\machs
\end{align*}
We set $\sigma$ to assign the small jobs according to $(x_{ij})$.
Since $\max_{j\in S} p_j\leq \delta B$ we get $\Cmax(\sigma) 	\leq \Cmax(\sigma')+\delta B$ and together with the above considerations:
\begin{equation*}
\Cmax(\sigma) \leq ((1+\delta)(1+3\delta)+2\delta)\Opt(I) \leq (1+\eps)\Opt(I)
\end{equation*}
\end{proof}
While all of the used techniques are well known they have---to the best of our knowledge---not been used in the indicated way up to now.

It can be easily seen that the class of instances of $\RA$, for which the rankwidth of the incidence graph is bounded by a constant $k$, is invariant with respect to changing the processing time of jobs and the introduction of copies of jobs.

\subsection*{Dynamic Program}

We present a dynamic program to solve $\RA$ using a branch decomposition $(T,\eta)$ with rankwidth $k$ for the incidence graph.
First we give some intuition on why a bounded rankwidth is useful.

For any Graph $(V,E)$ and $X\subseteq V$, we say that $u,v\in V$ have the same \emph{connection type with respect to $X$} if $N(u)\cap X= N(v)\cap X$. 
If $X$ is clear from the context we say that $u$ and $v$ have the same connection type.
Now, let $e=\set{a,b}\in E(T)$ be some edge of the branch decomposition and $\set{X_{e,a},X_{e,b}}$ the respective cut of $T$, i.e., $X_{e,x}$ for $x\in\set{a,b}$ is the set of vertices of $T$ that are in the same connected component as $x$ when the edge $e$ is removed.
Then $\set{X_{e,a},X_{e,b}}$ induces a partition of both the jobs and machines by $\cutjobs{e,x}:=\sett{j\in\jobs}{\eta(j)\in X_{e,x}}$ and $\cutmachs{e,x}:=\sett{i\in\machs}{\eta(j)\in X_{e,x}}$ for $x\in\set{a,b}$.
\begin{remark}\label{rem:rw_types_bounded_by_width}
Let $x,y\in\set{a,b}$ with $x\neq y$. The number of distinct connection types of $\cutjobs{e,x}$ with respect to $\cutmachs{e,y}$ is bounded by $2^k$.
\end{remark}
We actually use that the number of distinct connection types of the jobs is bounded.

In the rest of this section we first show how the branch decomposition can be used in a straightforward way to solve $\RA$ (with exponential running time).
The basic idea for this is that each edge of the decomposition corresponds to a partition of the job and machine sets and an optimal solution may be found by trying all possible ways of moving jobs between partitions.
At the machine-leafs all arriving jobs have to be processed, with no jobs going out, and at the job-leafs all jobs have to be send away, with no jobs coming in.
From this the procedure can work up to some root edge.
Next we argue that it is sufficient to consider only certain locally defined classes of job sets.
The crucial part here is that the number of these classes can be polynomially bounded, because the number of distinct sizes and connection types of jobs are constant.

\paragraph{Job sets.}

\begin{wrapfigure}{r}{1.7cm} 
\centering
\includegraphics{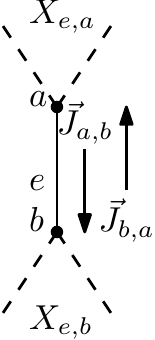}
\caption{Edge $e$}
\label{fig:rw_up_down}
\end{wrapfigure}
Let $e=\set{a,b}\in E(T)$ again be some edge of the tree $T$ and $\set{X_{e,a},X_{e,b}}$ the corresponding cut of $T$.
We fix a schedule $\sigma$ and make some basic observations.
There is a set of jobs $\sendjobs{a,b}\subseteq\cutjobs{e,a}$ that $\sigma$ assigns to machines from $\cutmachs{e,b}$.
We will use the intuition that $\sendjobs{a,b}$ is sent through $e$ from $a$ to $b$ (see also Figure \ref{fig:rw_up_down}).
The node $b$ may be an inner node or a leaf.
Moreover, if $b$ is a leaf, $\eta^{-1}(t)$ may be a job $j^*$ or a machine $i^*$.
In the first case $\sigma$ sends no jobs to $t$, and $j^*$ to $a$.
In the second case no jobs are sent to $a$ and the jobs send to $b$ should be feasible on $i^*$.
Now any set that is sent through an edge and arrives at an internal node, will be split into two parts: one going forth through the second and one through the third edge.
And looking at it the other way around:
Any set that is sent by a schedule through an edge coming from an inner node, is put together from two parts, one coming from the second and one coming from the third edge.

\begin{wrapfigure}{r}{2.3cm} 
\centering
\includegraphics{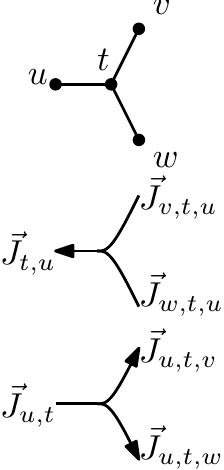}
\caption{Job-sets are split.}
\label{fig:rw_split}
\end{wrapfigure}
We formalize this notion.
Let $t$ be an internal node of $T$, with neighbours $u,v,w\in V(T)$.
Then for each pair of neighbours $x,y$ of $t$ there are job sets $\splitjobs{x,t,y}\subseteq \sendjobs{x,t}\cap \sendjobs{t,y}$, such that:
\begin{align}\label{eq:rw_job_splitting}
\sendjobs{u,t} &= \splitjobs{u,t,v} \dot{\cup} \splitjobs{u,t,w}  &   \sendjobs{t,u} &= \splitjobs{v,t,u} \dot{\cup} \splitjobs{w,t,u}
\end{align}
See also Figure \ref{fig:rw_split}.
It is rather easy to see that sets $\sendjobs{s,t}$ that are feasible at the leafs and fulfil the conditions (\ref{eq:rw_job_splitting}) uniquely define a feasible schedule. 

Using these observations, we now discuss how (the value of) an optimal schedule can be found by considering different job sets that may be sent through the edges.
For this we use an intuition of up and down, with $a$ above and $b$ below.
Let $\check{J}=\vec{J}_{a,b}\subseteq\cutjobs{e,a}$ and $\hat{J}=\vec{J}_{b,a}\subseteq\cutjobs{e,b}$ be some candidate sets to be sent up and down respectively through $e$.
We set $I_{e,x}(\hat{J},\check{J})=I[(\cutjobs{e,x}\setminus\vec{J}_{x,y})\cup\vec{J}_{y,x},\cutmachs{e,x}]$ for $x,y\in\set{a,b}$ with $x\neq y$, i.e., the subinstances of $I$ induced by $e$, if $\hat{J}$ and $\check{J}$ are send up or down respectively.
Note that the instance $I$ is split into the two subinstances.  
Moreover let $\Gamma_e$ be the set of pairs $(\hat{J},\check{J})$.
Then:
\begin{align}
\Opt(I) &= \min_{(\hat{J},\check{J})\in\Gamma_e}\max\set{\Opt(I_{e,a}(\hat{J},\check{J})),\Opt(I_{e,b}(\hat{J},\check{J}))} \label{eq:rw_opt_set_1}
\end{align}

We now consider the computation of $\Opt(I_{e,b}(\hat{J},\check{J}))$ for the two cases when $b$ is an internal node or a leaf.
If $b$ is a leaf, it may correspond to a job or a machine, i.e., $\eta^{-1}(b)=j^*$ or $\eta^{-1}(b)=i^*$.
In the first case we get $\Opt(I_{e,b}(\hat{J},\check{J}))=\infty$ if $\hat{J}\neq \set{j^*}=\cutjobs{e,b}$ or $\check{J}\neq \emptyset$, and $\Opt(I_{e,b}(\hat{J},\check{J}))=0$ otherwise.
In the second case $\hat{J}$ is empty since there are no jobs at $b$.
We get $\Opt(I_{e,b}(\hat{J},\check{J}))=\sum_{j\in\check{J}}p_j$ if $\check{J}\subseteq\validmachs{i^*}$ and $\Opt(I_{e,b}(\hat{J},\check{J}))=\infty$ otherwise.

\begin{wrapfigure}{r}{2.5cm} 
\centering
\includegraphics{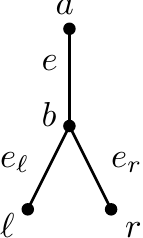}
\caption{Inner node $b$}
\label{fig:rw_left_right}
\end{wrapfigure}
Now let $b$ be an internal node that is connected to two lower nodes $\el$ and $\er$ via edges $e_{\el}$ and $e_{\er}$ (Figure \ref{fig:rw_left_right}).
We say that $\el$ and $e_{\el}$ are on the left, while $\er$ and $e_{\er}$ are on the right.
Recursively we assume that for any $(\hat{L},\check{L})\in\Gamma_{e_{\el}}$ and $(\hat{R},\check{R})\in\Gamma_{e_{\er}}$ we know $\Opt(I_{e_{\el},\el}(\hat{L},\check{L}))$ and $\Opt(I_{e_{\er},\er}(\hat{R},\check{R}))$ respectively.
We want to identify the set $\Lambda_e(\hat{J},\check{J})$ of tuples $(\hat{L},\check{L},\hat{R},\check{R})$ that for fixed $(\hat{J},\check{J})$ may occur in a schedule, i.e., fulfil condition (\ref{eq:rw_job_splitting}) for all edges from $\set{e,e_{\el},e_{\er}}$ .
For $\hat{J}$ it is clear which part is coming from the left and which from the right and we set $\hat{J}_{\el}\subseteq\cutjobs{e_{\el},\el}$ and $\hat{J}_{\er}\subseteq\cutjobs{e_{\er},\er}$ accordingly, such that $\hat{J}=\hat{J}_{\el}\dot{\cup}\hat{J}_{\er}$.
The other four sets in which the job sets going up and down could be split can all be tried.
More precisely, for each $\check{J}_{\el},\check{J}_{\er}\subseteq\check{J}$ with $\check{J}=\check{J}_{\el}\dot{\cup}\check{J}_{\er}$, $\hat{L}_{\er}\subseteq(\cutjobs{e_{\el},\el}\cap\jobs)\setminus\hat{J}_{\el}$ and $\hat{R}_{\el}\subseteq(\cutjobs{e_{\er},\er}\cap\jobs)\setminus\hat{J}_{\er}$ the tuple $(\hat{L}_{\er}\dot{\cup}\hat{J}_{\el},\check{J}_{\el}\dot{\cup}\hat{R}_{\el},\hat{R}_{\el}\dot{\cup}\hat{J}_{\er},\check{J}_{\er}\dot{\cup}\hat{L}_{\er})$ is in $\Lambda_e(\hat{J},\check{J})$ and the set is defined by such tuples.
We get:
\begin{align}
\Opt(I_{e,b}(\hat{J},\check{J})) &= \min_{(\hat{L},\check{L},\hat{R},\check{R})}\max\set{\Opt(I_{e_{\el},\el}(\hat{L},\check{L})),\Opt(I_{e_{\er},\er}(\hat{R},\check{R}))}\label{eq:rw_opt_set_2}
\end{align}
Using these considerations $\RA$ can be solved by choosing a root edge $e^*=\set{a^*,b^*}$ that is incident to a leaf $a^*$ corresponding to a job and designing a dynamic program working from leaf edges to the root edge using (\ref{eq:rw_opt_set_2}).
Now (\ref{eq:rw_opt_set_1}) for $e=e^*$ together with the considerations for leaf nodes yield $\Opt(I)=\Opt(I_{e^*,b^*}(\emptyset,\eta^{-1}(a^*)))$.
The running time of such an algorithm is exponential in the input length.

\paragraph{Classes of Jobs.}

Let $\check{J},\check{J}'\subseteq\cutjobs{e,a}$.
There are some cases in which $\check{J}$ and $\check{J}'$ are in some sense similar and it holds that $\Opt(I_{e,b}(\hat{J},\check{J})) = \Opt(I_{e,b}(\hat{J},\check{J}'))$.
This is the case if there is a bijection $\alpha:\check{J}\rightarrow\check{J}'$ such that $j$ and $\alpha(j)$ have the same connection type with respect to $\cutmachs{e,b}$ and $p_j=p_{\alpha(j)}$ for each $j\in\hat{J}$.
By this, an equivalence relation $\sim_{e,a}$ can be defined, and analogously a relation $\sim_{e,b}$.
Now the observation (\ref{eq:rw_opt_set_1}) can be reformulated in terms of equivalence classes:
\begin{lemma}\label{lem:rw_recurrence_1}
$\Opt(I) = \min_{([\hat{J}],[\check{J}])}\max\set{\min_{\hat{J}'}\Opt(I_{e,a}(\hat{J}',\check{J})), \min_{\check{J}'}\Opt(I_{e,b}(\hat{J},\check{J}'))}$.\qed
\end{lemma}

Note that in this equation equivalence classes $[\check{J}]$ and $[\hat{J}]$ are considered belonging to the relations $\sim_{e,a}$ and $\sim_{e,b}$ respectively.
$\hat{J}$ and $\check{J}$ are arbitrary representatives of these classes.
We will now develop a sensible representation for the equivalence classes.

We drop the notion of up and down for the following considerations, i.e., $b\in e$ is just one of two nodes of some edge $e$.
We assume some ordering of the different processing times, with $p(i)$ denoting the $i$-th processing time for $i\in [d]$.
Any set of jobs $J'$ induces a vector $\lambda\in\mathbb{Z}^{d}_{\geq 0}$ where $\lambda_i$ is the number of jobs in $J'$ that have the $i$-th processing time, i.e., $\lambda_i=|\sett{j\in J'}{p_j=p(i)}|$.
We set $p(\lambda)=\sum_{i\in[d]}p(i)\lambda_i$.
Let $\kappa(e,b)$ be the number of connection types of jobs from $\cutjobs{e,b}$ with respect to $\cutmachs{e,a}$.
Note that due to Remark \ref{rem:rw_types_bounded_by_width} we get $\kappa(e,b)\leq 2^k$.
Again assuming some ordering, for $i\in[\kappa(e,b)]$ let $\typesizes{e,b}(i)$ be the size vector induced by the $i$-th connection type of $\cutjobs{e,b}$ with respect to $\cutmachs{e,a}$ and $\typemachs{e,a}(i)\subseteq\cutmachs{e,a}$ the machines from $\cutmachs{e,a}$ the respective jobs may be processed on.
Moreover let $\typesizes{e,b}=(\typesizes{e,b}(1),\dots,\typesizes{e,b}(\kappa(e,b)))$.
Now the equivalence classes of $\sim_{e,b}$ can naturally be represented and characterized by vectors $\iota\leq\typesizes{e,b}$.
\begin{remark}\label{rem:rw_classes_bounded_by_k_and_d}
For each $e\in E(T)$ and $b\in e$ there are at most $n^{\kappa(e,b)d}$ different vectors $\iota\leq\typesizes{e,b}$.
\end{remark}

We now study the splitting behaviour of job classes at inner nodes.
Consider a set $J'$ that is sent through an edge $f=\set{u,v}$ and then forth through an incident edge $g=\set{v,w}$.
Then there are vectors $\iota_f$ and $\iota_g$ representing $J'$ in the context of $f$ and $g$ respectively.
However, some other set $J''$ represented by $\iota_f$ will also be represented by $\iota_g$, that is $\iota_f$ translates uniquely into $\iota_g$.
We formalize this notion by the definition of a translation function $\tau_{f,g}:\sett{\iota}{\iota\leq\typesizes{f,u}}\rightarrow\sett{\iota'}{\iota'\leq\typesizes{g,v}}$. 
For each $i\in[\kappa(f,u)]$ there is a unique $i'\in[\kappa(g,v)]$ with $\typemachs{f,v}(i)\cap\cutmachs{g,w}=\typemachs{g,w}(i')$, i.e., the $i$-th connection type of $\cutjobs{f,u}$ translates into the $i'$-th connection type of $\cutjobs{g,v}$. 
Let $\upsilon_{f,g}:[\kappa(f,u)]\rightarrow[\kappa(g,v)]$ be given by $i\mapsto i'$.
Now for each $\iota\leq\typesizes{f,u}$ we set $\tau_{f,g}(\iota)= (\iota'(1),\dots,\iota'(\kappa(g,v)))$, with $\iota'(i')\in\mathbb{Z}_{\geq 0}^d$ and more precisely $\iota'(i')=\sum_{i\in\upsilon^{-1}_{f,g}(i')}\iota(i)$.
With this we can formulate an analogue of (\ref{eq:rw_job_splitting}).
For a fix schedule $\sigma$ let $\iota_{a,b}$ be the representative of the set of jobs that $\sigma$ sends from $a$ to $b$.
Now let $t$ be an inner node with neighbours $u,v,w$.
For neighbours $x,y$ of $t$, the set $\sendjobs{x,t,y}$ considered in the last paragraph now has a representative in the contexts of $\set{x,t}$ and $\set{t,y}$.
Fixing the first one $\iota_{x,t,y}\leq\iota_{x,t}$, the second one can be obtained via a transformation function, yielding:
\begin{align}\label{eq:rw_job_class_split}
\iota_{u,t} &= \iota_{u,t,v} + \iota_{u,t,w} & \iota_{t,u} &= \tau_{(\set{v,t},\set{u,t})}(\iota_{v,t,u}) + \tau_{(\set{w,t},\set{u,t})}(\iota_{w,t,u})
\end{align} 

From now on we will use the up and down notion like before ($e=\set{a,b}\in E(T)$ with $a$ above and $b$ below).
Let $\hat{\iota}\leq\typesizes{e,b}$ and $\check{\iota}\leq \typesizes{e,a}$ be candidate job classes to be send up and down $e$.
Considering Lemma \ref{lem:rw_recurrence_1} we set $\Opt(e,\hat{\iota},\check{\iota})$ to be the minimum value $\Opt(I_{e,t}(\hat{J},\check{J}))$ where $\hat{J}$ is represented by $\hat{\iota}$ and $\check{J}$ by $\check{\iota}$.

For the case when $b$ is a leaf not much changes.
If $\eta^{-1}(b)$ is a job $j^*$, the class of $\set{j^*}$ has only one element and is represented by $\typesizes{e,b}$.
Therefore we get that $\Opt(e,\hat{\iota},\check{\iota})=0$ for $\hat{\iota}=\typesizes{e,b}$ and $\check{\iota}=0$, and $\Opt(e,\hat{\iota},\check{\iota})=\infty$ otherwise.
If $\eta^{-1}(b)$ is a machine $i^*$, we have $\typesizes{e,b}=0$ and there are only two possible connection types for jobs from $\cutjobs{e,a}$, because jobs can be processed on $i^*$ or not, i.e., $\kappa(e,a)\leq 2$.
In any case we get $\Opt(e,\hat{\iota},\check{\iota})=\sum_{i\in[\kappa(e,a)]}p(\check{\iota}(i))$.

Now let $b$ be an inner node again with lower neighbours $\el$ and $\er$ to which it is connected via edges $e_{\el}$ and $e_{\er}$.
We may assume that we know the values $\Opt(e_{\el},\hat{\lambda},\check{\lambda})$ and $\Opt(e_{\er},\hat{\rho},\check{\rho})$ for candidate job classes $(\hat{\lambda},\check{\lambda},\hat{\rho},\check{\rho})$ to go up or down the left or right edge respectively.
We want to identify the set $\Xi_{e}(\hat{\iota},\check{\iota})$ of quadruples $(\hat{\lambda},\check{\lambda},\hat{\rho},\check{\rho})$ that are compatible with $\hat{\iota}$ and $\check{\iota}$, i.e., that fulfil (\ref{eq:rw_job_class_split}).
For this let $\check{\iota}_{\el},\check{\iota}_{\er}\leq\check{\iota}$ with $\check{\iota}_{\el}+\check{\iota}_{\er}=\check{\iota}$, $\hat{\lambda}_{\el},\hat{\lambda}_{\er}\leq\typesizes{e_{\el},\el}$ with $\hat{\lambda}_{\el}+\hat{\lambda}_{\er}\leq\typesizes{e_{\el},\el}$, and $\hat{\rho}_{\el},\hat{\rho}_{\er}\leq\typesizes{e_{\er},\er}$ with $\hat{\rho}_{\el}+\hat{\rho}_{\er}\leq\typesizes{e_{\er},\er}$, such that $\tau_{e_{\el},e}(\hat{\lambda}_{\el})+\tau_{e_{\er},e}(\hat{\rho}_{\er}) = \hat{\iota}$. 
By setting $\hat{\lambda} = \hat{\lambda}_{\el}+\hat{\lambda}_{\er}$, $\check{\lambda}=\tau_{e,e_{\el}}(\check{\iota}_{\el})+\tau_{e_{\er},e_{\el}}(\hat{\rho}_{\el})$, $\hat{\rho} = \hat{\rho}_{\el}+\hat{\rho}_{\er}$ and $\check{\rho}=\tau_{e,e_{\er}}(\check{\iota}_{\er})+\tau_{e_{\el},e_{\er}}(\hat{\lambda}_{\er}))$ we get such a tuple and the set $\Xi_{e}(\hat{\iota},\check{\iota})$ is defined by such tuples.
\begin{lemma}\label{lem:rw_splitting}
$\Opt(e,\hat{\iota},\check{\iota}) = \min_{(\hat{\lambda},\check{\lambda},\hat{\rho},\check{\rho})}
\max\set{\Opt(e_{\el},\hat{\lambda},\check{\lambda}),\Opt(e_{\er},\hat{\rho},\check{\rho})}$.
\end{lemma}
\begin{proof}
If the righthand side equals $\infty$, it is easy to see that the equation holds, and we will therefore assume $\Opt(e,\hat{\iota},\check{\iota})<\infty$.
For given $\hat{\iota}\leq\typesizes{e,a}$ and $\check{\iota}\leq\typesizes{e,b}$ let $\check{J}\subseteq\cutjobs{e,a}$ be any set represented by $\check{\iota}$ and $\hat{J}^*\subseteq\cutjobs{e,b}$ an optimal one represented by $\hat{\iota}$, i.e., minimizing  $\Opt(I_{e,b}(\hat{J},\check{J}))$.
Let $\sigma^*$ be an optimal schedule for $I_{e,b}(\hat{J},\check{J})$.
Furthermore let $\hat{L}^*$, $\check{L}^*$, $\hat{R}^*$ and $\check{R}^*$ be the sets that $\sigma^*$ sends up or down through $e_{\el}$ or $e_{\er}$ respectively, and let $\hat{\lambda}^*$, $\check{\lambda}^*$, $\hat{\rho}^*$ and $\check{\rho}^*$ their classes.
Than $\sigma^*$ induces schedules $\sigma_{\el}^*$ and $\sigma_{er}^*$ for $I_{e,\el}$ and $I_{e,\er}$. 
We get:
\begin{align*}
\Cmax(\sigma^*) 	&= \max\set{\Cmax(\sigma_{\el}^*),\Cmax(\sigma_{\er}^*)}\\
				&\geq \max\set{\Opt(e_{\el},\hat{\lambda}^*,\check{\lambda}^*),\Opt(e_{\er},\hat{\rho}^*,\check{\rho}^*)}\\
				&\geq \min_{\substack{(\hat{\lambda},\check{\lambda},\hat{\rho},\check{\rho})\in\\ \Xi_{e}(\hat{\iota},\check{\iota})}}
\max\set{\Opt(e_{\el},\hat{\lambda},\check{\lambda}),\Opt(e_{\er},\hat{\rho},\check{\rho})}
\end{align*}
Now let $(\hat{\lambda},\check{\lambda},\hat{\rho},\check{\rho})\in \Xi_{e}(\hat{\iota},\check{\iota})$ minimizing the right hand side of of the considered equation with corresponding splitting vectors $\check{\iota}_{\el},\check{\iota}_{\er},\hat{\lambda}_{\el},\hat{\lambda}_{\er},\hat{\rho}_{\el},\hat{\rho}_{\er}$.
Moreover let $\hat{L}$ and $\hat{R}$ be optimal sets represented by $\hat{\lambda}$ and $\hat{\rho}$ respectively. 
Splitting $\hat{L}$, $\hat{R}$ and $\check{J}$ corresponding to the splitting of their job classes we can obtain sets $\check{L}$, $\check{R}$ and $\hat{J}$ that are represented by $\check{\lambda}$, $\check{\rho}$ and $\hat{\iota}$ respectively and fulfil (\ref{eq:rw_job_splitting}).
We have now $\Opt(e_{\el},\hat{\lambda},\check{\lambda})=\Opt(I_{e_{\el},\el}(\hat{L},\check{L}))$ and $\Opt(e_{\er},\hat{\rho},\check{\rho})=\Opt(I_{e_{\er},\er}(\hat{R},\check{R}))$.
Let $\sigma_{\el}$ and $\sigma_{\er}$ be respective optimal schedules.
Than $\sigma:=\sigma_{\el}\cup\sigma_{\er}$ is a schedule for $I_{e,b}(\hat{J},\check{J})$ and $\Cmax(\sigma)$ equals the right hand side of the considered equation. 
Since $\sigma^*$ was chosen optimal with an optimal class representative we have furthermore $\Cmax(\sigma)\geq \Cmax(\sigma^*)$.
This yields:
\begin{align*}
\Opt(e,\hat{\iota},\check{\iota}) 	&= \Cmax(\sigma^*) =  \Cmax(\sigma)\\
									&= \min_{\substack{(\hat{\lambda},\check{\lambda},\hat{\rho},\check{\rho})\in\\ \Xi_{e}(\hat{\iota},\check{\iota})}}
\max\set{\Opt(e_{\el},\hat{\lambda},\check{\lambda}),\Opt(e_{\er},\hat{\rho},\check{\rho})}
\end{align*}
Moreover we get that $\hat{J}$ is optimal as well.
\end{proof}

\paragraph{Results.}

With these considerations a dynamic program for $\RA$ can be defined.
This can be done in a way such that its running time is in $\Oh(m^2n^{\Oh(d2^k)})$, proving Theorem \ref{thm:rw_result} together with the Rounding Lemma and the considerations of Section \ref{sec:prelim}.

\subsection*{Bi-cographs}

We show that the path- and tree-hierarchical and nested cases are all special cases of the case that the incidence graph is a bi-cograph.
Bi-cographs were introduced as a bipartite analogue of cographs \cite{GV97}.
\begin{definition}
For a bipartite graph $G=(A\dot{\cup}B,E)$ the \emph{bi-complement} of $G$ is the graph $(A\dot{\cup}B,\sett{\set{a,b}}{a\in A, b\in B, \set{a,b}\not\in E})$. 
A graph is called \emph{bi-cograph}, iff it is bipartite and can be reduced to isolated vertices by recursively bi-complementing its connected bipartite subgraphs. 
\end{definition}
It is known \cite{GV00} that their cliquewidth and therefore also their rankwidth is bounded by $4$.
Furthermore by recursively bi-complementing the connected bipartite subgraphs, a certain decomposition of a given bi-cograph can be found in linear time that is similar to cotrees of cographs \cite{GV97}.
This decomposition can easily be turned into a branch decomposition, for which in the application studied here the number of connection types of jobs $\kappa(e,u)$ for every edge $e$ of the decomposition and $v\in e$ is bounded by $2$.
\begin{lemma}
Let $I$ be an instance of $\RA$ with path- or tree-hierarchical or nested restrictions.
Then the incidence graph of $I$ is a bi-cograph.
\end{lemma}
\begin{proof}
We first consider the case that $I$ has tree hierarchical restrictions.
Let $T$ be a corresponding rooted tree with $V(T)=\machs$. 
Then there is at least one machine (the root of $T$) that can process all jobs.
After bi-complementing the connected bipartite subgraphs of the incidence graph this machine is isolated.
This can be repeated:
After bi-complementing two more times the nearest descendants of the root in $T$ that cannot process all jobs will be isolated. 
Iterating this, at some point all machines and therefore also all jobs will be isolated.

Now let $I$ be an instance with nested restrictions.
Note that the jobs $j\in\jobs$ with maximal $\validmachs{j}$ (with respect to $\subseteq$) are all in different connected components of the incidence graph and connected to all machines in their component.
Hence they are isolated after bi-complementing the first time.
If we bi-complement a second time and remove these jobs we get a new instance with nested restrictions and less jobs.
By iterating this argument the claim follows.
\end{proof}

\paragraph{Acknowledgements.}

The Rounding Lemma in the presented form was formulated by Lars Rohwedder and Kevin Prohn as part of a student project.

\bibliography{library}

\end{document}